\documentclass[a4paper, 10pt]{article}
\usepackage{graphicx}
\usepackage{amsmath}
\usepackage{amsthm}
\usepackage{amssymb}
\usepackage{latexsym}
\usepackage{dsfont}

\newtheorem{theorem}{Theorem}[section]
\newtheorem{definition}[theorem]{Definition}
\newtheorem{lemma}[theorem]{Lemma}

\newtheorem{corollary}[theorem]{Corollary}

\def\wt{\mathrm{wt}}

\title{On the nonexistence of $\left[\binom{2m}{m-1}, 2m, \binom{2m-1}{m-1}\right]$, $m$ odd, complex orthogonal design}
\author{Yuan Li and Haibin Kan}
\date{August 15, 2011}

\begin{document}

\maketitle

\begin{abstract}
Complex orthogonal designs (CODs) are used to construct space-time
block codes. COD $\mathcal{O}_z$ with parameter $[p, n, k]$ is a $p
\times n$ matrix, where nonzero entries are filled by $\pm z_i$ or
$\pm z^*_i$, $i = 1, 2, \ldots, k$, such that $\mathcal{O}^H_z
\mathcal{O}_z = (|z_1|^2+|z_2|^2+\ldots+|z_k|^2)I_{n \times n}$.
Adams et al. in ``The final case of the decoding delay problem for
maximum rate complex orthogonal designs,'' IEEE Trans. Inf. Theory,
vol. 56, no. 1, pp. 103-122, Jan. 2010, first proved the
nonexistence of $\left[\binom{2m}{m-1}, 2m,
\binom{2m-1}{m-1}\right]$, $m$ odd, COD. Combining with the previous
result that decoding delay should be an integer multiple of
$\binom{2m}{m-1}$, they solved the final case $n \equiv 2 \pmod 4$
of the decoding delay problem for maximum rate complex orthogonal
designs.

In this paper, we give another proof of the nonexistence of COD with
parameter $\left[\binom{2m}{m-1}, 2m, \binom{2m-1}{m-1}\right]$, $m$
odd. Our new proof is based on the uniqueness of
$\left[\binom{2m}{m-1}, 2m-1, \binom{2m-1}{m-1}\right]$ under
equivalence operation, where an explicit-form representation is
proposed to help the proof. Then, by proving it's impossible to add
an extra orthogonal column on COD $\left[\binom{2m}{m-1}, 2m-1,
\binom{2m-1}{m-1}\right]$ when $m$ is odd, we complete the proof of
the nonexistence of COD $\left[\binom{2m}{m-1}, 2m,
\binom{2m-1}{m-1}\right]$.

\textbf{Key words:} complex orthogonal design, space-time block
codes, maximal rate and minimal delay.
\end{abstract}

\section{Introduction}
Space-time block codes have been widely investigated for wireless
communication systems with multiple transmit and receive antennas.
Since the pioneering work by Alamouti \cite{Ala98} in 1998, and the
work by Tarokh et al. \cite{TJC99}, \cite{TJC00}, orthogonal designs
have become an effective technique for the design of space-time
block codes (STBC). The importance of this class of codes comes from
the fact that they achieve full diversity and have the fast
maximum-likelihood (ML) decoding.

A complex orthogonal design (COD) $\mathcal{O}_z[p, n, k]$ is an $p
\times n$ matrix, and each entry is filled by $\pm z_i$ or $\pm
z^*_i$, $i=1,2,\ldots, k$, such that $\mathcal{O}^H_z \mathcal{O}_z
= \sum_{i=1}^{n}|z_i|^2 I_n$, where $H$ is the Hermitian transpose
and $I_n$ is the $n \times n$ identity matrix. Under this
definition, the designs are said to be combinatorial, in the sense
that there is no linear processing in each entry. Code rate $k/p$
and decoding delay $p$ are the two most important criteria of
complex orthogonal space-time block codes. One important problem is,
given $n$, determine the tight upper bound of code rate, which is
called maximal rate problem. Another is, given $n$, determine the
tight lower bound of decoding delay $p$ when code rate $k/p$ reaches
the maximal, which is called minimal delay problem.

For combinatorial CODs, where linear combination is not allowed,
Liang determined for a COD with $n=2m$ or $2m-1$, the maximal
possible rate is $\frac{m+1}{2m}$ \cite{Lia03}. Liang gave an
algorithm in \cite{Lia03} to generate such CODs with rate
$\frac{m+1}{2m}$, which shows that this bound is tight. The minimal
delay problem are solved by Adams et al. In \cite{AKP07}, lower
bound ${2m \choose m-1}$ of decoding delay is proved for any $n=2m$
or $2m-1$. And further, it's proved that the decoding delay must be
a multiple of ${2m \choose m-1}$. In \cite{AKM10}, by showing the
nonexistence of COD $\left[\binom{2m}{m-1}, 2m,
\binom{2m-1}{m-1}\right]$ with even $m$, Adams et al. prove that
when $n \equiv 2 \pmod 4$, decoding delay $p$ is lowered bound by
$2{2m \choose m-1}$.

The tightness of above upper bound of rate is shown by constructions
in \cite{Lia03}, \cite{SXL04}, \cite{LFX05}. When $n \not \equiv 0
\pmod 4$, those in \cite{Lia03}, \cite{SXL04} also achieve the lower
bound of delay. And constructions in \cite{LFX05} achieves minimal
delay for all $n$.

 The organization of our paper is as follows. In
section 2, we introduce the notions, definitions and some known
results which will be used. In section 3, we present our main
results including the uniqueness of COD with parameter
$\left[\binom{2m}{m-1}, 2m-1, \binom{2m-1}{m-1}\right]$, and the
nonexistence of COD having parameter $\left[\binom{2m}{m-1}, 2m,
\binom{2m-1}{m-1}\right]$ which depends on the former result. In
order to prove the main results, an explicit-form construction of
optimal COD is introduced, which is crucial to our proofs.

\section{Preliminaries}
In this section, we introduce some basic notions, which will be used
in the sequel.

$\mathds{C}$ denotes the field of complex numbers, $\mathds{R}$ the
field of real numbers and $\mathbb{F}_2$ the field with two
elements. Adding over $\mathbb{F}_2$ is denoted by $\oplus$ to avoid
ambiguity. All vectors are assumed to be column vectors. For any
field $\mathds{F}$, denoted by $\mathds{F}^n$ and $M_{m \times
n}(\mathds{F})$ the set of all $n$-dimensional vectors in
$\mathds{F}$ and the set of all $m \times n$ matrices in
$\mathds{F}$, respectively. In this paper, rows and variables are
often indexed by vectors in $\mathbb{F}^2_n$.

For convenience, let $e_i \in \mathbb{F}^n_2$ be the vector with
$i^{\text{th}}$ bit occupied by $1$ and the others $0$, i.e., $ e_i
= (\underbrace{0, \ldots, 0}_{i-1}, 1, \underbrace{0, \ldots,
0}_{n-i}) $ and let $e = e_1 \oplus e_2 \oplus \ldots \oplus e_n$,
i.e., $ e = (1, 1, \ldots, 1)_2. $ The weight of a vector in
$\mathbb{F}^n_2$ is defined as the number of ones in $n$ bits, i.e.,
$\wt(\alpha) = \sum_{i=1}^{n} \alpha(i)$. Furthermore, $\wt_{s,
t}(\alpha)$ is defined as the sum of $s^{\text{th}}$ bit to
$t^{\text{th}}$ bit, i.e.,
\begin{displaymath}
\wt_{s, t}(\alpha) = \alpha(s) + \alpha(s+1) + \cdots + \alpha(t) =
\sum_{i = s}^{t}{\alpha(i)}.
\end{displaymath}

In abuse of notation, we denote by $z[j]$ the complex variable
$z_j$, up to negation and conjugation, i.e., $z[j] \in \{ z_j, -z_j,
z^*_j, -z^*_j\}.$ Note that the same notation $z[j]$ may represent
different elements in the same paragraph.

\begin{definition}
A $[p, n, k]$ complex orthogonal design $\mathcal{O}_z$ is a $p
\times n$ rectangular matrix whose nonzero entries are
$
z_1, z_2, \ldots, z_{k}, -z_1, -z_2, \ldots, -z_{k}
$
or their conjugates
$
z_1^*, z_2^*, \ldots, z_{k}^*, -z_1^*, -z_2^*, \ldots, -z_{k}^*,
$
where $z_1, z_2, \ldots, z_{k}$ are indeterminates over
 $\mathds{C}$, such that
\begin{displaymath}
\mathcal{O}_z^H\mathcal{O}_z = (|z_1|^2+|z_2|^2+\cdots+|z_{k}|^2)
I_{n \times n}.
\end{displaymath}
$k/p$ is called the code rate of $\mathcal{O}_z$, and $p$ is called
the decoding delay of $\mathcal{O}_z$.
\end{definition}

A matrix is called an Alamouti $2 \times 2$ if it matches the
following form
\begin{equation}
\label{equ:Alamouti}
\begin{pmatrix}
z_i & z_j \\
-z_j^* & z_i^*
\end{pmatrix},
\end{equation}
up to negation or conjugation of $z_i$ or $z_j$. We say two rows
share an Alamout $2 \times 2$ if and only if the intersection of the
two rows and some two columns form an Alamouti $2 \times 2$.

\begin{definition}
The equivalence operations performed on any COD are defined as
follows.
\begin{itemize}
\item[] 1) Rearrange the order the rows(``row permutation'').
\item[] 2) Rearrange the order the columns (``column permutation'').
\item[] 3) Conjugate all instances of certain variable (``instance
conjugation'').
\item[] 4) Negate all instances of certain variable (``instance
negation'').
\item[] 5) Change the index of all instances of certain variable (``instance
renaming'').
\item[] 6) Multiply any row by $-1$, (``row negation'').
\item[] 7) Multiply any column by $-1$, (``column negation'').
\end{itemize}
\end{definition}

It's not difficult to verify that, given a COD
$\mathcal{O}_z[p,n,k]$, after arbitrary equivalence operations, we
will obtain another COD $\mathcal{O}'_z[p,n,k]$. And we say COD
$\mathcal{O}_z$ and $\mathcal{O}'_z$ are the same under equivalence
operations.

Following the definition in \cite{Lia03}, with a little
modification, define an $(n_1, n_2)$-$\mathcal{B}_j$ form by
\begin{eqnarray}
\mathcal{B}_j & = & \begin{pmatrix} z_j I_{n_1} & \mathcal{M}_1 \\
-\mathcal{M}^H_1 & z^*_j I_{n_2}
\end{pmatrix} \nonumber\\
& = &
       \left(
       \begin{array}{c|c}
       \begin{array}{cccc}
       z_j&0&\cdots&0\\
       0&z_j&\cdots&0\\
       \vdots&\vdots&\ddots&\vdots\\
       0&0&\cdots&z_j
       \end{array}&\mathcal{M}_j\\
       \hline
       -\mathcal{M}_j^H&\begin{array}{cccc}
       z_j^*&0&\cdots&0\\
       0&z_j^*&\cdots&0\\
       \vdots&\vdots&\ddots&\vdots\\
       0&0&\cdots&z_j^*
       \end{array}
       \end{array}
       \right), \label{equ:Bj_form}
\end{eqnarray}
where $n_1 + n_2 = n$. And we call it $\mathcal{B}_j$ form for
short.

\begin{definition} \cite{AKP07}
We say COD $\mathcal{O}_z$ is in $\mathcal{B}_j$ form if the
submatrix $\mathcal{B}_j$ can be created from $\mathcal{O}_z$
through equivalence operations except for column permutation.
Equivalently, $\mathcal{O}_z$ is in $\mathcal{B}_j$ form if every
row of $\mathcal{B}_j$ appears within the rows of $\mathcal{O}_z$,
up to possible conjugations of all instances of $z_i$ and possible
factors of $-1$.
\end{definition}

 It is proved that
\cite{AKP07} that COD $\mathcal{O}_z$ is in some $\mathcal{B}_j$
form if and only if one row in $\mathcal{O}_z$ matches one row of
$\mathcal{B}_j$ up to signs and conjugations.

In \cite{Lia03}, Liang proved the upper bound $\frac{m+1}{2m}$ of
code rate $\frac{k}{p}$ for any $n=2m$ or $2m-1$, and obtained the
necessary and sufficient condition to reach the maximal rate.

\begin{theorem}
\label{thm:max_rate} Let $n=2m$ or $2m-1$. The rate of COD
$\mathcal{O}_z[p,n,k]$ is upper bounded by $\frac{m+1}{2m}$, i.e.,
$\frac{k}{p} \le \frac{m+1}{2m}$.

This bound is achieved if and only if for all $i=1,2,\ldots, k$,
$\mathcal{B}_j$ is an $(m,m-1)$-$\mathcal{B}_j$ or
$(m-1,m)$-$\mathcal{B}_j$ form and there are no zero entries in
$\mathcal{M}_j$, when $n=2m-1$; $\mathcal{B}_j$ is an
$(m,m)$-$\mathcal{B}_j$ form and there are no zero entries in
$\mathcal{M}_j$, when $n=2m$.
\end{theorem}

The lower bound on the decoding delay when code rate reaches the
maximal is completely solved by Adams et al. in \cite{AKP07} and
\cite{AKM10}.

\begin{theorem} Let $n=2m$ or $2m-1$. For COD
$\mathcal{O}_z[p,n,k]$, if the rate reaches the maximal, i.e.,
$\frac{k}{p} = \frac{m+1}{2m}$, the delay, i.e., $p$, is lower
bounded by ${2m \choose m-1}$ when $n \equiv 0,1,3 \pmod 4$; by
$2{2m \choose m-1}$ when $n \equiv 2 \pmod 4$.
\end{theorem}

The technique in proving the lower bound ${2m \choose m-1}$ is the
observation and definition of zero pattern, which is a vector in
$\mathbb{F}^n_2$ defined with respect to one row where the
$i^{\text{th}}$ bit is $0$ if and only if the element on column $i$
is $0$. For example, when
\begin{equation}
\mathcal{O}_z=
\begin{pmatrix}
z_1 & z_2 & z_3 \\
-z^*_2 & z^*_1 & 0 \\
-z^*_3 & 0 & z^*_1 \\
0 & z^*_3 & -z^*_2
\end{pmatrix},
\label{equ:COD433}
\end{equation}
the first row has zero pattern $(1,1,1)$, the second $(1,1,0)$, the
third $(1,0,1)$, the fourth $(0,1,1)$.

In \cite{AKP07}, it's proved that the decoding delay is an integer
multiple of ${2m \choose m-1}$. Therefore, in order to prove the
lower bound of delay $2{2m \choose m-1}$ for $n \equiv 2 \pmod 4$,
it's sufficient to prove the nonexistence of $\left[{2m \choose
m-1}, 2m, {2m-1 \choose 2m-1}\right]$. The basic idea in
\cite{AKM10} is, first proving the uniqueness of COD with parameter
$\left[{2m \choose m-1}, 2m-1, {2m-1 \choose 2m-1}\right]$ under
equivalence operation, then showing is impossible to add an extra
column for a specific COD with parameter $\left[{2m \choose m-1},
2m-1, {2m-1 \choose 2m-1}\right]$ to obtain a new one. Our proof
follows the same basic idea, but different from theirs, we define an
explicit-form COD $\mathcal{G}_{2m-1}\left[{2m \choose m-1}, 2m-1,
{2m-1 \choose 2m-1}\right]$, while another standard from is defined
to help prove the uniqueness in \cite{AKM10}. Due to our
explicit-form construction, it's much easier to show the
impossibility of adding an extra orthogonal column to
$\mathcal{G}_{2m-1}$.

\section{Main Results}

The following lemma is first proved in \cite{AKM10}, which an
observation of $\mathcal{B}_j$ form.

\begin{lemma}
\label{lem_zp_mode} For maximal rate COD $\mathcal{O}_z[p,n,k]$, if
$|\mathcal{O}_z(\alpha, i)|=|\mathcal{O}_z(\beta,j)|$ then the zero
patterns of row $\alpha$ and row $\beta$ are only different at
column $i$ and $j$; if $|\mathcal{O}_z(\alpha,
i)|=|\mathcal{O}_z(\beta,j)|^*$ then the zero pattern of row
$\alpha$ and row $\beta$ are only the same at column $i$ and $j$.
\end{lemma}

The following lemma is also first proved in \cite{AKP07}. For
completeness, we give another proof.

\begin{lemma}
\label{lem_zp_exist} For maximal rate COD $\mathcal{O}_z[p,n,k]$,
$n=2m$ or $2m-1$, then $p \ge \binom{2m}{m-1}$. When $n=2m$, every
zero pattern with $m+1$ ones exist; when $n=2m-1$, every zero
pattern with $m$ or $m+1$ ones exist.
\end{lemma}
\begin{proof}
First, we will prove if one zero pattern of some row is $\alpha \in
\mathbb{F}_2^n$, then for any $\alpha(i)=1, \alpha(j)=0$, there
exists one row with zero pattern $\beta \in \mathbb{F}_2^n$, such
that $\beta(i) = \alpha(j), \beta(j) = \alpha(i)$ and $\beta(l) =
\alpha(l)$ for all $l \not= i, j$. To see the existence of zero
pattern $\beta$, we only need to arrange $\mathcal{O}_z$ into
$\mathcal{B}_{\gamma}$ form, where $\mathcal{O}_z(\alpha, i)=
z[\gamma]$.

Then, since any permutation is a product of transpositions, all zero
patterns with weight $m$ (or $m+1$) exists.
\end{proof}

As a consequence of Lemma \ref{lem_zp_mode} and Lemma
\ref{lem_zp_exist}, we know, up to negations, CODs with parameter
$\left[\binom{2m}{m-1}, 2m-1, \binom{2m-1}{m-1}\right]$ are the same
under equivalence operation. The following lemma is to help define
the ``standard form'' COD $\mathcal{G}_{2m-1}\left[\binom{2m}{m-1},
2m-1, \binom{2m-1}{m-1}\right]$.

\begin{lemma} \cite{AKM11} \label{lem_conj_separated} Let $\mathcal{O}_z$ be a maximum rate, minimal delay COD with parameter $\left[\binom{2m}{m-1},
2m-1, \binom{2m-1}{m-1}\right]$. Then $\mathcal{O}_z$ is equivalent
to a COD that is conjugation separated, where the rows containing
$m$ nonzero entries are all conjugated, and those containing $m+1$
nonzero entries are all non-conjugated.
\end{lemma}

 By Lemma \ref{lem_conj_separated}, we
know $\mathcal{O}_z$ can be made conjugation separated with rows
containing $m$ nonzero entries all conjugated. We identify each row
by its zero pattern and with the ${2m}^{\text{th}}$ bit denoting
whether the row is conjugated or not, i.e., $\alpha \in
\mathbb{F}_2^{2m}$ with $\wt(\alpha)=m+1$.

Let $\mathcal{G}_{2m-1}$ be a COD with parameter
$\left[\binom{2m}{m-1}, 2m-1, \binom{2m-1}{m-1}\right]$ with rows
identified by vectors in $\mathbb{F}^{2m}_2$ with weight $m+1$ and
columns identified by $1, 2, \ldots, 2m-1$. The elements of
$\mathcal{G}_{2m-1}$ are determined by the following rules.
\begin{itemize}
\item If $\alpha(i)=1$ and $\alpha(2m)=0$, $\mathcal{O}_z(\alpha,
i)=(-1)^{\theta(\alpha,i)} z_{\alpha\oplus e_i}$;
\item
If $\alpha(i)=1$ and $\alpha(2m)=1$, $\mathcal{O}_z(\alpha,
i)=(-1)^{\theta(\alpha,i)} z^*_{\alpha\oplus e_i\oplus e}$;
\item If $\alpha(i)=0$, $\mathcal{O}_z(\alpha, i)=0$.
\end{itemize}
Here
\begin{equation}
\label{equ:def_theta} \theta(\alpha, i) =
\begin{cases}
wt_{i, 2m}(\alpha) + \frac{i}{2}, & \text{if $i$ is even}, \\
wt_{i, 2m}(\alpha) + \frac{i-1}{2} + \alpha(2m), & \text{if $i$ is
odd}.
\end{cases}
\end{equation}

From Lemma \ref{lem_zp_mode}, we know rows $\alpha, \beta \in
\mathbb{F}^{2m}$ share an Alamouti $2 \times 2$ if and only if their
zero patterns are only the same at column $1 \le i, j \le 2m-1$,
such that $\alpha(i)=\alpha(j)=\beta(i)=\beta(j)=1$ and $\alpha
\oplus \beta = e \oplus e_i \oplus e_j$. Submatrix
$\mathcal{G}_{2m-1}(\alpha,\beta; i, j)$ has the following form
$$
\begin{pmatrix}
(-1)^{\theta(\alpha, i)}z_{\alpha \oplus e_i} & (-1)^{\theta(\alpha, j)}z_{\alpha \oplus e_j}\\
(-1)^{\theta(\beta, i)}z^*_{\beta \oplus e_i \oplus e} &
(-1)^{\theta(\beta, j)}z^*_{\beta \oplus e_j \oplus e}
\end{pmatrix}
$$
or
$$
\begin{pmatrix}
(-1)^{\theta(\alpha, i)}z^*_{\alpha \oplus e_i} & (-1)^{\theta(\alpha, j)}z^*_{\alpha \oplus e_j}\\
(-1)^{\theta(\beta, i)}z_{\beta \oplus e_i \oplus e} &
(-1)^{\theta(\beta, j)}z_{\beta \oplus e_j \oplus e}
\end{pmatrix}.
$$
Note that $\alpha \oplus \beta = e \oplus e_i \oplus e_j \Rightarrow
\alpha \oplus e_i = \beta \oplus e_j \oplus e$. Thus, we only need
to check the signs to see whether submatrix
$\mathcal{G}_{2m-1}(\alpha,\beta; i, j)$ is an Alamouti $2 \times
2$.

Let's calculate $\theta(\alpha, i) + \theta(\beta, i)$ according to
the parity of $i$, by definition \eqref{equ:def_theta}. When $i$ is
even, $\theta(\alpha, i) + \theta(\beta, i)=
\wt_{i,2m}(\alpha)+\frac{i}{2} + \wt_{i, 2m}(\beta) + \frac{i}{2}
\equiv \wt_{i,2m}(\alpha \oplus \beta) + i \pmod 2$; When $i$ is
odd, $\theta(\alpha, i) + \theta(\beta, i)=
\wt_{i,2m}(\alpha)+\frac{i-1}{2} + \alpha(2m) + \wt_{i, 2m}(\beta) +
\frac{i-1}{2} + \beta(2m) \equiv \wt_{i,2m}(\alpha \oplus \beta) + i
\pmod 2$. Therefore, we see
\begin{equation}
\label{equ:theta_ab_sum} \theta(\alpha, i) + \theta(\beta, i) \equiv
\wt_{i,2m}(\alpha \oplus \beta) + i \pmod 2
 \end{equation} always holds.

Then,
\begin{eqnarray}
& & \theta(\alpha, i) + \theta(\beta, i)+\theta(\alpha, j) +
\theta(\beta, j) \nonumber\\
& \equiv & \wt_{i,2m}(\alpha \oplus \beta) + i + \wt_{j,2m}(\alpha
\oplus \beta) + j \nonumber\\
& \equiv & \wt_{i, j}(\alpha \oplus \beta) + i + j \nonumber\\
& \equiv & j-i-1+i+j \nonumber\\
& \equiv & 1 \pmod 2. \nonumber
\end{eqnarray}
In the last second step, $\wt_{i, j}(\alpha \oplus \beta) = j-i-1$
is true because $\alpha \oplus \beta = e \oplus e_i \oplus e_j$.

Up to now, by constructing a specific function $\theta(\alpha, i)$,
we see $\mathcal{G}_{2m-1}$ is a COD. In fact, we only need to know
such arrangement of signs exists.

\begin{theorem}
\label{thm:2m-1_eqval} Let $\mathcal{O}_z$ be a COD with parameter
$\left[\binom{2m}{m-1}, 2m-1, \binom{2m-1}{m-1}\right]$. Then
$\mathcal{O}_z$ is equivalent to $\mathcal{G}_{2m-1}$ under
equivalence operation.
\end{theorem}
\begin{proof} The basic idea is to show $\mathcal{O}_z$ and
$\mathcal{G}_{2m-1}$ can be transformed into a standard form COD.
Since equivalence operations are invertible, we claim
$\mathcal{O}_z$ is equivalent to $\mathcal{G}_{2m-1}$.

Before defining standard form, we first introduce a total order on
vectors of $\mathbb{F}_2^{2m}$, that $\alpha < \beta \Leftrightarrow
\sum_{i=1}^{2m}{\alpha(i) 2^i} < \sum_{i=1}^{2m}{\beta(i) 2^i}$.
Now, we will show $\mathcal{O}_z$, as well as $\mathcal{G}_{2m-1}$,
can be transformed into a ``standard from'' uniquely. Consider
variable $z[\gamma]$ by increasing order of $\gamma$, where $\gamma
\in \mathbb{F}_2^{2m}$, $\wt(\gamma)=2m$ and $\gamma(2m)= 0$. Keep
in mind that our algorithm determines the signs of all instance of
$z[\gamma]$ once at a time, and once the signs are determined, it
will never change in subsequent steps. Take out all rows containing
$z[\gamma]$, which is
\begin{eqnarray}
\mathcal{B}_\gamma & = & \begin{pmatrix} \pm z_\gamma I_{n_1} & \mathcal{M}_\gamma \\
\mathcal{M}'_\gamma & \pm z^*_\gamma I_{n_2}
\end{pmatrix} \nonumber\\
& = &
       \left(
       \begin{array}{c|c}
       \begin{array}{cccc}
       \pm z_\gamma&0&\cdots&0\\
       0& \pm z_\gamma&\cdots&0\\
       \vdots&\vdots&\ddots&\vdots\\
       0&0&\cdots& \pm z_\gamma
       \end{array}&\mathcal{M}_\gamma\\
       \hline
       \mathcal{M}'_\gamma &\begin{array}{cccc}
       \pm z_\gamma^*&0&\cdots&0\\
       0& \pm z_\gamma^*&\cdots&0\\
       \vdots&\vdots&\ddots&\vdots\\
       0&0&\cdots& \pm z_\gamma^*
       \end{array}
       \end{array}
       \right), \label{equ:Bj_form}
\end{eqnarray}
where $n_1, n_2 \in \{m, m-1\}$ and $n_1 + n_2 = 2m-1$.

Followings are two steps of our algorithm
\begin{itemize}
\item For those $z[\gamma]$
whose index are not smallest on the corresponding row, we will show
there are only two possible ways to determine their signs. In other
words, their relationships, same or opposite, are fixed due to the
determined signs of $z[\delta]$, where $\delta < \gamma$. At last,
we make use of instance negation to make sure, on the smallest row,
$z[\gamma]$ is positive.
\item If, in one row, $z[\gamma]$ is the element with the smallest index,
which implies all other elements in the same row are undetermined,
we can use row negation to make sure it's positive without affecting
the determined signs.
\end{itemize}

Now, we prove the claim that ``there are only two possible ways to
determine their signs'' in the first step above is true. Let
$\gamma(1)=\gamma(2)=\ldots=\gamma(s)=1, \gamma(s+1)=0,
\gamma(t-1)=1, \gamma(t)=\gamma(t+1)=\ldots=\gamma(2m)=0,$ where $0
\le s < t \le 2m$. And assume $\mathcal{O}_z(\alpha_i, i) =
z[\gamma]$ for $1 \le i \le 2m-1$, which implies $\alpha_i = \gamma
\oplus e_i$ when $\gamma(i)=0$, $\alpha_i = \gamma \oplus e_i \oplus
e$ when $\gamma(i)=1$.

For any $1 \le i\not=j \le 2m-1$, consider the element in
$\mathcal{O}_z(\alpha_i, j)$. When $\alpha_i(j)=0 \Leftrightarrow
\gamma(i)=\gamma(j)$, $\mathcal{O}_z(\alpha_i, j)=0$ by definition.
When $\alpha_i(j)=1 \Leftrightarrow \gamma(i)\oplus\gamma(j)=1$,
$\mathcal{O}_z(\alpha_i, j)=z[\delta]$, where $\delta=\alpha_i
\oplus e_j \oplus \alpha_i(2m)e=\gamma \oplus e_i \oplus \gamma(i)e
\oplus e_j \oplus \gamma(i)e  = \gamma \oplus e_i \oplus e_j$.
Therefore, for $1\le i\le s$ or $t \le i \le 2m-1$,
$\mathcal{O}_z(\alpha_i, i)$ is the element with the smallest index
on that row.

For $s < i < t$ and $1 \le j \le 2m-1$ satisfying $\gamma(i) \oplus
\gamma(j)=1$, submatrix $$\begin{pmatrix} \mathcal{O}_z(\alpha_i, i)
& \mathcal{O}_z(\alpha_i, j) \\ \mathcal{O}_z(\alpha_j, i) &
\mathcal{O}_z(\alpha_j, j)
\end{pmatrix} = \begin{pmatrix} z[\gamma] & z[\gamma \oplus e_i \oplus e_j]\\
z[\gamma \oplus e_i \oplus e_j] & z[\gamma]\end{pmatrix}$$ is an
Alamouti $2 \times 2$. As our algorithm determines the signs of
$z[\gamma]$ by increasing order, the signs of $z[\gamma \oplus e_i
\oplus e_j]$s are determined if and only if $\gamma\oplus e_i \oplus
e_j < \gamma \Leftrightarrow \gamma(i)=1, \gamma(j)=0, i > j$ or
$\gamma(i)=0, \gamma(j)=1, i < j$. Therefore, if $\gamma(i)=0$, the
relationship of signs of $\mathcal{O}_z(\alpha_i, i)$ and
$\mathcal{O}_z(\alpha_{t-1}, t-1)$ are determined; if $\gamma(i)=1$,
the relationship of signs of $\mathcal{O}_z(\alpha_i, i)$ and
$\mathcal{O}_z(\alpha_{s+1}, s+1)$ are determined. Since the
relationship of signs of $\mathcal{O}_z(\alpha_{s+1}, s+1)$ and
$\mathcal{O}_z(\alpha_{t-1}, t-1)$ are determined, and by the
transitivity of sign relationship, we claim all
$\mathcal{O}_z(\alpha_i, i)$ for $s < i < t$ are uniquely
determined.
\end{proof}

It's worth noting that in the proof of Theorem \ref{thm:2m-1_eqval},
``column negation'' operation is not used. Therefore, any COD with
parameter $\left[\binom{2m}{m-1}, 2m-1, \binom{2m-1}{m-1}\right]$
can be transformed into $\mathcal{G}_{2m-1}$ without using ``column
negation'' operation.

\begin{lemma}
\label{lem:pad_imposs}
 When $m$ is odd, it's impossible to obtain a COD with parameter
$\left[\binom{2m}{m-1}, 2m, \binom{2m-1}{m-1}\right]$ by adding an
extra column on $\mathcal{G}_{2m-1}$.
\end{lemma}
\begin{proof} Assume that there exists such a $\left[\binom{2m}{m-1}, 2m, \binom{2m-1}{m-1}\right]$ COD by adding an extra
column on $\mathcal{G}_{2m-1}$. Denote the last column by
$\mathcal{L}_{2m}$, and assume $\mathcal{O}_z = \left(
\mathcal{G}_{2m-1}, \mathcal{L}_{2m} \right)$ is a
$\left[\binom{2m}{m-1}, 2m, \binom{2m-1}{m-1}\right]$ COD.

By Lemma \ref{lem_zp_mode}, we know $\mathcal{L}_{2m}$, up to
negations, are uniquely determined. It's not difficult to verify
that $\mathcal{L}_{2m}(\alpha) = \alpha(2m) (-1)^{\phi(\alpha)}
z_{\alpha \oplus e_{2m}}$ for $\alpha \in \mathbb{F}_2^{2m}$ and
$\wt(\alpha)=m+1$, where $\phi(\alpha) \in \mathbb{F}_2$ are
undetermined.

For any $\alpha \in \mathbb{F}_2^{2m}$ with $\wt(\alpha)=m+1$ and
$\alpha(2m)=1$, $\mathcal{O}_z(\alpha, i)$ and
$\mathcal{O}_z(\alpha, 2m)$ are contained in the following Alamouti
$2 \times 2$
$$
\begin{pmatrix}
\mathcal{O}_z(\alpha, i) & \mathcal{O}_z(\alpha, 2m) \\
\mathcal{O}_z(\beta, i) & \mathcal{O}_z(\beta, 2m)
\end{pmatrix} =
\begin{pmatrix}
(-1)^{\theta(\alpha, i)} z^*_{\alpha \oplus e_i \oplus e} &
(-1)^{\phi(\alpha)} z_{\alpha \oplus e_{2m}} \\
(-1)^{\theta(\beta, i)} z^*_{\beta \oplus e_i \oplus e} &
(-1)^{\phi(\beta)} z_{\beta \oplus e_{2m}}
\end{pmatrix},
$$
where $\beta = \alpha \oplus e_i \oplus e_{2m} \oplus e$. For
$\alpha \oplus \beta = e_i \oplus e_{2m} \oplus e$, calculate
$\theta(\alpha, i) + \theta(\beta, i)$ by definition
\eqref{equ:def_theta}. \textbf{When $i$ is even,} $\theta(\alpha, i)
+ \theta(\beta, i) \equiv \wt_{i,2m}(e_i \oplus e_{2m} \oplus e) + i
\equiv  2m-i-1 + i\equiv 1 \pmod 2$, which implies $\phi(\alpha)=
\phi(\alpha \oplus e_i \oplus e_{2m} \oplus e)$ for even $i$.
\textbf{When $i$ is odd,} $\theta(\alpha, i) + \theta(\beta, i)
\equiv \wt_{i,2m}(e_i \oplus e_{2m} \oplus e) + (i-1) + 2 \equiv
(2m-i-1) + (i-1) + 2\equiv 0 \pmod 2$, which implies $\phi(\alpha)=
\phi(\alpha \oplus e_i \oplus e_{2m} \oplus e) \oplus 1$ for odd
$i$.

Now, we are ready to induce the contradiction. For any $\alpha \in
\mathbb{F}_2^n$, let $i = 2l$, $l = 1, 2, \ldots, m-1$ and $i =
2l-1$, $l = 1, 2, \ldots, m$ separately. We have
\begin{eqnarray*}
\phi(\alpha)  &= & \phi\left(\alpha \bigoplus_{l=1}^{m-1}({e_{2l}
\oplus e_{2m} \oplus e}) \bigoplus_{l=1}^{m}(e_{2l-1} \oplus e_{2m}
\oplus
e)\right) \oplus m\\
& = & \phi\left(\alpha \oplus e_{2m} \oplus e
\bigoplus_{l=1}^{2m-1}{e_{l}}\right) \oplus 1\\
& = & \phi(\alpha \oplus e_{2m} \oplus e \oplus e \oplus e_{2m}) \oplus 1\\
& = & \phi(\alpha) \oplus 1,
\end{eqnarray*}
which is a contradiction!
\end{proof}

Equipped with the above results, we are able to prove the
unexistence of COD with parameter $\left[\binom{2m}{m-1}, 2m,
\binom{2m-1}{m-1}\right]$, $m$ odd.

\begin{theorem} There does not exist COD with
parameter $\left[\binom{2m}{m-1}, 2m, \binom{2m-1}{m-1}\right]$ when
$m$ is odd.
\end{theorem}
\begin{proof} We prove it by contradiction. Assume there exists a COD $\mathcal{O}_z$ with parameter $\left[\binom{2m}{m-1}, 2m,
\binom{2m-1}{m-1}\right]$. Deleting one column, we obtain a COD with
parameter $\left[\binom{2m}{m-1}, 2m-1, \binom{2m-1}{m-1}\right]$,
which is denoted by $\mathcal{O}'_z$. By Theorem
\ref{thm:2m-1_eqval}, we know $\mathcal{O}'_z$ can be obtained by
equivalence operation over $\mathcal{G}_{2m-1}$. Since equivalence
operation is invertible, apply the inverse operation on
$\mathcal{O}_z$, we obtain a COD $\left( \mathcal{G}_{2m-1},
\mathcal{L}_{2m} \right)$. By Lemma \ref{lem:pad_imposs}, we know
it's impossible to add an extra column on $\mathcal{G}_{2m-1}$ still
to be orthogonal.
\end{proof}

The following corollary is a direct consequence of the previous
results.

\begin{corollary} When $n \equiv 0, 1, 3 \pmod 4$, CODs with parameter $[p, n, k]$ achieving maximal
rate and minimal delay are the same under equivalence operation.
\end{corollary}
\begin{proof} When $n \equiv 1, 3 \pmod 4$, i.e., $n=2m-1$ for integer $m$, COD $\mathcal{O}_z$ with
parameter $\left[{2m \choose m-1}, 2m-1, {2m-1 \choose m-1}\right]$
achieves maximal rate and minimal delay. By Theorem
\ref{thm:2m-1_eqval}, we know $\mathcal{O}_z$ is the same as
$\mathcal{G}_{2m-1}$ under equivalence operation.

When $n \equiv 0 \pmod 4$, i.e., $n = 2m$, $m$ even, COD
$\mathcal{O}_z$ with parameter $\left[{2m \choose m-1}, 2m, {2m-1
\choose m-1}\right]$ achieves maximal rate and minimal delay. Since
by deleting one column of $\mathcal{O}_z$ we obtain a maximal-rate,
minimal-delay COD for $n=2m-1$, which is equivalent to
$\mathcal{G}_{2m-1}$ by Lemma \ref{thm:2m-1_eqval}. By Lemma
\ref{lem_zp_mode}, we know the remaining column is uniquely
determined regardless of signs.

Following the argument in Lemma \ref{lem:pad_imposs}, it's
sufficient to prove function $\phi(\alpha)$, $\alpha \in
\mathbb{F}_2^{2m}, \wt(\alpha)=m+1$, is uniquely determined up to a
negation of all. From the proof of Lemma \ref{lem:pad_imposs}, we
know $\phi(\alpha)=\phi(\alpha\oplus e_i \oplus e_{2m} \oplus e)$
for even $i$; and $\phi(\alpha)=\phi(\alpha\oplus e_i \oplus e_{2m}
\oplus e) \oplus 1$ for odd $i$, where $\alpha(i)=1$. Again, take
integer $j$ such that $(\alpha\oplus e_i \oplus e_{2m} \oplus
e)(j)=1 \Leftrightarrow \alpha(j)=0$, we can obtain the relationship
between $\phi(\alpha)$ and $\phi(\alpha\oplus e_i \oplus e_{2m}
\oplus e \oplus e_j \oplus e_{2m} \oplus e)=\phi(\alpha\oplus e_i
\oplus e_j)$. Since $i, j$ are taken arbitrarily, we know all
relationships between $\phi(\alpha)$ are determined.

\end{proof}

\section{Acknowledgment}
We are immensely grateful to Chen Yuan's help for deciding the sign
function in \eqref{equ:def_theta}.

\end{document}